\pdfoutput=1
\documentclass[letterpaper, 10 pt, conference]{ieeeconf}  
\usepackage{amsmath}
\usepackage{amssymb}
\usepackage{algorithm}
\usepackage{tikz}
\usepackage{hyperref}
\usepackage{url}
\usepackage[noend]{algpseudocode}
\usepackage{graphicx}
\usepackage{pgfplots}
\pgfplotsset{compat=1.11}
\usepackage{mathtools}
\usepackage[noadjust]{cite}
\usepgfplotslibrary{fillbetween} 
\usepackage{subcaption}
\usepackage{marvosym}
\usepackage{fontawesome5}
\usepackage{centernot}

\usepackage[colorinlistoftodos,prependcaption]{todonotes}
\IEEEoverridecommandlockouts                              
\algrenewcommand\textproc{\textsc}
\DeclareMathOperator*{\argmax}{arg\,max}
\DeclareMathOperator*{\argmin}{arg\,min}
\DeclareMathOperator{\sgn}{sgn}
\newcommand{\gradxhat}{K^T \nabla h_S(\hat{x})}

\title{Stealthy Deactivation of Safety Filters}

\author{Daniel Arnström, André M.H. Teixeira%
\thanks{D. Arnstr\"om and A. Teixeira are with the Division of Systems and Control, Depratment of Information Technology, Uppsala University, Sweden 
{\tt\small \{daniel.arnstrom,andre.teixeira\}@it.uu.se}
This work is supported by the Swedish Foundation for Strategic Research.
}%
}

\begin{document}
\renewcommand{\baselinestretch}{1.0}

\definecolor{set19c1}{HTML}{E41A1C}
\definecolor{set19c2}{HTML}{377EB8}
\definecolor{set19c3}{HTML}{4DAF4A}
\definecolor{set19c4}{HTML}{984EA3}
\definecolor{set19c5}{HTML}{FF7F00}
\definecolor{set19c6}{HTML}{FFFF33}
\definecolor{set19c7}{HTML}{A65628}
\definecolor{set19c8}{HTML}{F781BF}
\definecolor{set19c9}{HTML}{999999}

\maketitle
\thispagestyle{empty}
\pagestyle{empty}
\newtheorem{proposition}{Proposition}
\newtheorem{lemma}{Lemma}
\newtheorem{corollary}{Corollary}
\newtheorem{remark}{Remark}
\newtheorem{theorem}{Theorem}
\newtheorem{definition}{Definition}
\newtheorem{assumption}{Assumption}
\newtheorem{example}{Example}
\newtheorem{problem}{Problem}

\pgfplotstableread{data/new/gradient1dof.dat}{\gradonedofnew}
\pgfplotstableread{data/new/gradient2dof.dat}{\gradtwodofnew}
\pgfplotstableread{data/new/none.dat}{\nonadversnew}
\pgfplotstableread{data/new/random.dat}{\randomnew}

\begin{abstract}
 Safety filters ensure that only safe control actions are executed. We propose a simple and stealthy false-data injection attack for deactivating such safety filters; in particular, we focus on deactivating safety filters that are based on control-barrier functions. The attack injects false sensor measurements to bias state estimates to the interior of a safety region, which makes the safety filter accept unsafe control actions. To detect such attacks, we also propose a detector that detects biases manufactured by the proposed attack policy, which complements conventional detectors when safety filters are used. The proposed attack policy and detector are illustrated on a double integrator example.

\end{abstract}

\section{Introduction}

Cyber-physical systems (CPSs) integrate communication, computation, and control technologies, which enable advanced control systems that are efficient, sustainable, and resilient \cite{dibaji2019systems}. The cyber components of CPSs do, however, also open up for new vulnerabilities since they enable cyber attacks \cite{teixeira2015secure}\cite[\S 4.C]{roadmap}. To address these novel vulnerabilities, several control-theoretic approaches that analyze and improve the security of CPSs have been proposed \cite{chong2019tutorial}. Such approaches are especially important for CPSs that are \textit{safety critical}, since successful attacks to such systems can have severe, even fatal, consequences \cite[\S 4.B]{roadmap}.

In parallel with the development of the above-mentioned work on security for CSPs, a promising framework has been developed for safety of control systems through so-called \textit{safety filters} \cite{wabersich2023data,tomlin2003computational,ames2017cbf,wabersich2018linear}. 
A safety filter takes in a desired control action and the current state of the system, and outputs a filtered control action that guarantees a safe behaviour of the system. Such filters separate safety from performance, since any controller can be used to produce the desired control action; this allows for safety guarantees even when using unpredictable controllers such as experimental data-driven controllers \cite{coulson2019data}.

In this paper we consider a cyber attack that injects false data on the communication channel from sensor measurements to a state observer, as illustrated in Figure \ref{fig:overview}. The goal of the attack is to produce synthetic measurements $y^a$ that ``deactivate'' the safety filter, which in turn allows for dangerous control actions to be applied to the plant. Our specific focus in this paper is on safety filters that are based on control-barrier functions (CBFs), but the main idea applies to other types of safety filters (see, e.g., \cite{wabersich2023data} for a recent survey on safety filters). 

A common protection against false-data injections is anomaly detectors \cite{murguia2019model}, which raise an alarm if the obtained measurement is too far from the expected measurement. Here we specifically consider \textit{stealthy} attacks, which are designed to circumvent such anomaly detectors \cite{teixeira2015secure}. 

The proposed stealthy false-data injection attack biases the state estimates toward the center of a safe set, making the safety filter accept unsafe control actions.  An important difference with the proposed attack and classical false-data injection attacks (see, e.g., \cite{mo2010false}) is that the adversary does not need direct access to a dynamical model of the system under attack, which is due to the attack being tailored to deactivate safety filters.
The contributions of this work can be interpreted as extending some of the work on adversarial examples from a static setting \cite{goodfellow2015adversarial} to a dynamic setting.

In addition to an attack that biases state estimates towards the center of a safe set, we also propose a way of detecting such biases. To this end, we proposed a detector that correlates the difference of expected and actual measurements with directions towards the interior of the safe set.  

\begin{figure}
    \centering
    \begin{tikzpicture}[scale=0.8, transform shape]
    \usetikzlibrary{shapes,arrows}
    \usetikzlibrary{calc}
    \tikzstyle{block} = [
    rectangle, draw, fill=white!20, 
    text width=5.25em, text centered, 
    rounded corners, minimum height=3em]

    \node[block] (detect) at (0,-2.5) {\faBell \:Detector};
    \node[block] (test1) at (0,1.05) {\faGamepad\:Primary \\Controller};
    \node[block] (test2) at (3.75,0) {\faShield* Safety \\ Filter};
    \node[block] (test3) at (7,0) {\faPlane\:Plant};
    \node[block] (observ) at (0,-1.05) {\faBinoculars\:Observer};
    \node (spy) at (4,-2.3) {\huge \faUserSecret};
    \node[red,right of=observ,xshift=1.5em,yshift=-1em] (adv-y) {$y^a$};
    \path [draw, -latex'] (test1) -- ++(2.4,0)  node[above]{$u_{\text{des}}$}-|(test2);
    \path [draw, -latex'] (test2) ++(1.6,0) |- ($(observ.east)+(0,+0.1)$);
    \path [draw, -latex'] (test2) --node[above]{$u_{\text{act}}$} (test3);
    \path [draw, latex'-] ($(observ.east)+(0,-0.1)$) -- ++(3,0) -| node[above right,yshift=0.5em]{$y$} (test3);
    \path [draw, -latex'] (observ) --node[left]{$\hat{x}$} (test1);
    \path [draw, -latex'] (observ) |- (test2);
    \path [draw, -latex'] (observ) --node[left]{$r$} (detect);
    \path[draw,-latex,red] (spy) -- node[right]{\Lightning} ++(0,1);

\end{tikzpicture}
    \caption{\small Overview of the system architecture considered in this paper. The safety filter produces a safe control action $u_{\text{act}}$ given a desired control $u_{\text{des}}$ and the current estimated state $\hat{x}$. An adversary tries to deactivate this filter through false-data injections on the communication channel between the sensors and the observer by replacing the true measurement $y$ with a synthetic measurement $y^a$.}
    \label{fig:overview}
\end{figure}
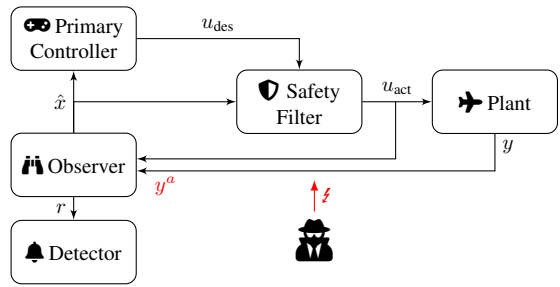

To summarize, the main contributions of the paper are: 
\begin{enumerate}
  \item A stealthy false-data injection attack that biases state estimates to deactivate CBF-based safety filters.
  \item An anomaly detector that detects measurement residuals that are biased toward the interior of the safe set. 
\end{enumerate}
\subsection{Definitions and notation}
The gradient of a function $f(x)$ is denoted $\nabla f(x) \triangleq \frac{\partial f}{\partial x}$ and is represented as a column vector. The $i$th element of a vector $v$ is denoted $[v]_i$.
The boundary of a set $S$ is denoted $\partial S$.
A continuous function $\alpha : (-a,b) \to \mathbb{R}$, with $a,b>0$, is said to be extended class $\mathcal{K}$ ($\alpha \in \mathcal{K}^e$) if it is strictly increasing and $\alpha(0) = 0$.
\section{Preliminaries and problem formulation}
We consider dynamical systems of the form 
\begin{equation}
    \label{eq:sys}
    \dot{x} = f(x,u),
\end{equation}
with state $x \in \mathbb{R}^{n_x}$ and control $u  \subseteq \mathbb{R}^{n_u}$. The dynamics of the system is given by $f : \mathbb{R}^{n_x}\times \mathbb{R}^{n_u} \to \mathbb{R}^{n_x}$.
Moreover, the system generates measurements $y\in \mathbb{R}^{n_y}$ according to
\begin{equation}
 y = h(x) + e,
\end{equation}
with the measurement function $h: \mathbb{R}^{n_x} \to \mathbb{R}^{n_y}$ and measurement noise $e$ (for our purpose, we make no particular assumptions on $e$.)

An estimate $\hat{x}\in \mathcal{X}$ of the system state $x$ is obtained by an observer of the form 
\begin{equation}
    \label{eq:xhat-dot}
    \dot{\hat{x}} = f(\hat{x},u)+ K (y-h(\hat{x})),
\end{equation}
where the observer gain $K \in \mathbb{R}^{n_x \times n_y}$ might be time- and state-dependent. In addition to correcting the state estimate, the innovation $y-h(\hat{x})$ from the observer is used as the residual $r$ in an anomaly detector; in particular, an anomaly is flagged if $\|y-h(\hat{x})\| > \delta $, where $\delta > 0$ is a user-specified threshold that trades off sensitivity and false-detections.

\subsection{Safety and Invariance}
In this paper we are interested in the \textit{safety} of system \eqref{eq:sys}, defined through admissible states and control actions.
\begin{definition}[Safety]
    Given a set of admissible states $\mathcal{X}\subseteq \mathbb{R}^{n_x}$ and a set of admissible controls $\mathcal{U}\subseteq \mathbb{R}^{n_u}$, the system in \eqref{eq:sys} is said to be \textit{safe} if 
    \begin{equation}
        x(t) \in \mathcal{X} \text{ and } u(t) \in \mathcal{U} \text{ for all } t \geq 0.
    \end{equation}
\end{definition}

A common way to ensure safety is to find a \textit{forward invariant} subset of the admissible states. 
To define forward invariance, let $\kappa : \mathcal{X} \to \mathcal{U}$ be a control policy, which results in the closed-loop (autonomous) system 
\begin{equation}
    \label{eq:sys-closed}
    \dot{x} = f(x,\kappa(x)).
\end{equation}

For an autonomous system of the form \eqref{eq:sys-closed} we define forward invariance in the following way.

\begin{definition}[Forward invariance]
    A set $\mathcal{S} \subseteq \mathbb{R}^{n_x}$ is \textit{forward invariant} for the closed-loop system \eqref{eq:sys-closed} if $x(0) \in \mathcal{S}$ implies that $x(t) \in \mathcal{S}$ for all $t\geq0$.   
\end{definition}

Next, we make the connection between forward invariant sets and safety of \eqref{eq:sys} explicit. 
\begin{lemma}[Forward invariance $\rightarrow$ safety]
    The system in \eqref{eq:sys} is safe if there exists a control policy $\kappa : \mathcal{X} \to \mathcal{U}$ and a set $S\subseteq \mathcal{X}$ such that $S$ is forward invariant for the closed loop system \eqref{eq:sys-closed}. 
\end{lemma}
\begin{proof}
    From the range of $\kappa$ we directly have that $u(t) \in \mathcal{U}$ for all $t\geq 0$ since $u(t) = \kappa(x)$. Moreover, we have from the forward invariance of $S$ that $x(t)\in S$ for all $t\geq 0$, but since $S \subseteq \mathcal{X}$ we also have that $x(t) \in \mathcal{X}$ for all $t\geq0$. 
\end{proof}

We will call invariant sets such that $\mathcal{S} \subseteq \mathcal{X}$ \textit{safe sets}; additionally, we assume that safe sets can be characterized with a continuously differentiable function $h_S: \mathcal{X}\to \mathbb{R}$ as 
\begin{equation}
    \label{eq:Sdef}
    \mathcal{S}=\{x\in \mathbb{R}^{n_x} : h_S(x) \geq 0\}.
\end{equation}

\subsection{Safety filters}
In the previous subsection we established the connection between the safety of \eqref{eq:sys} and the existence of a forward invariant set. To be able to use this in practice, however, a more pragmatic characterization of forward invariance is necessary. 
Most safety filters build on the following classical characterization of forward invariance:
\begin{theorem}[Nagumo's theorem \cite{nagumo1942lage}]
    \label{th:nagumo}
    Let $S \subseteq \mathcal{X}$ be defined as $S\triangleq\{x \in \mathbb{R}^{n_x} | h_S(x) \geq 0\}$, where $h_S$ is continuously differentiable. Moreover, assume that the interior of $S$ is non-empty. Then $S$ is forward invariant for \eqref{eq:sys-closed} if and only if 
\begin{equation}
    \label{eq:nagumo}
    \dot{h}_S(x) = \nabla h_S(x)^T f(x, \kappa(x)) \geq 0
\end{equation}
for any $x\in \partial S$.
\end{theorem}

Intuitively, Theorem \ref{th:nagumo} states that forward invariance of $S$ is equivalent to the system state changing towards the interior of $S$ when it is at the boundary $\partial S$.   
The condition \eqref{eq:nagumo} is, however, not practical since it is only enforced on the boundary $\partial S$, which has measure zero. One motivation behind introducing \textit{control-barrier functions} (CBFs) is to obtain a Nagumo-like condition that holds on the entire $S$. In particular, this is done through adding a strengthening term in the form of an extended class $\mathcal{K}$ function to the right-hand side of \eqref{eq:nagumo}. 
\begin{definition}[Control Barrier Function]
    The function $h_S$ is a control-barrier function for \eqref{eq:sys} if there exists $\alpha \in \mathcal{K}^e$ such that  
    \begin{equation}
        \sup_{u\in \mathcal{U}} \nabla h_S(x)^T f(x, u) \geq -\alpha(h_S(x)).
    \end{equation}
    for any $x \in S$.
\end{definition}

A key relationship between CBFs and forward invariance is stated in the following theorem.  
\begin{theorem}[CBF $\rightarrow$ forward invariance \cite{ames2017cbf}]
    Let $S \subseteq \mathcal{X}$ be defined as $S\triangleq\{x \in \mathbb{R}^{n_x} | h_S(x) \geq 0\}$, where $h_S$ is a continuously differentiable CBF. Moreover, assume that $\nabla h_S(x) \neq 0$ for any $x\in\partial S$. Then any Lipschitz continuous control policy $\kappa: \mathcal{X} \to \mathcal{U}$ such that $\kappa(x)\in \tilde{\mathcal{U}}_S(x)$ makes $S$ forward invariant for the closed-loop system \eqref{eq:sys-closed}.  
\end{theorem}

We introduce notation for the corresponding set of control actions $\tilde{\mathcal{U}}_S(x)$ that ensure forward invariance of $S$ given the state $x$; that is, 
\begin{equation}
    \label{eq:cbf-U}
    \tilde{\mathcal{U}}_S(x) \triangleq \{u\in \mathcal{U} : \nabla h_S(x)^T f(x, u) \geq -\alpha(h_S(x)) \}.
\end{equation}

A safety filter that is based on CBFs is the \textit{Active Set Invariance Filter} (ASIF) \cite{ames2017cbf} that filters control actions according to the policy 
\begin{equation}
    \label{eq:qp-cbf}
    \begin{aligned}
        \kappa_{\text{ASIF}}(x,u_{\text{des}}) = &\argmin_{u\in \mathcal{U}} \|u-u_{\text{des}}\|^2 \\
                                   &\nabla h_S(x)^T f(x, u) \geq - \alpha(h_S(x)), 
    \end{aligned}
\end{equation}
where $u_{\text{des}}\in \mathbb{R}^{n_u}$ is a desired control action. If the set of admissible control actions $\mathcal{U}$ is a polyhedron and the dynamics $f$ is control affine, the resulting optimization problem in \eqref{eq:qp-cbf} is a quadratic program (QP). The ASIF is said to be \textit{minimally invasive}, since it finds a safe control action that is as close as possible to the desired control action $u_{\text{des}}$. In the rest of the paper, we will consider how safety filters of the form \eqref{eq:qp-cbf} can be ``deactivated''.

\subsection{Formulation of adversarial problem}
The overarching objective of the adversary is to make the system unsafe; that is, to get $x \notin \mathcal{X}$. Since we assume that control actions applied to the system are filtered through a safety filter of the form \eqref{eq:qp-cbf}, which ensures that $x \in \mathcal{X}$, an attacker \textit{cannot} make the system unsafe as long as the state $x$ is estimated correctly. To this end, the adversary need to alter the measurement $y$ to bias the state estimate $\hat{x}$ such that the safety filter becomes ``deactivated''. We define \textit{deactivation} of a CBF (and the corresponding ASIF given by that CBF) in the following way.
\begin{definition}[Deactivation]
    The CBF is said to be \textit{deactivated} by the state estimate $\hat{x}$ at time $t$ if $\exists u \in \tilde{\mathcal{U}}(\hat{x}(t))$ such that $u \notin \tilde{\mathcal{U}}(x(t))$.  
\end{definition}

In other words, a state estimate $\hat{x}$ deactivates a safety filter if there exist a control action that is deemed safe for $\hat{x}$ but unsafe for the true state $x$. When a safety filter is deactivated, there might exist a control signal that breaks the forward invariance of $S$ and, consequently, makes the system \eqref{eq:sys} unsafe. A delimitation we make in this paper is, however, that we are only interested in deactivating safety filters, not producing unsafe control actions, as stated in the following remark. 

\begin{remark}
    In this paper our focus is \textit{not} on producing a control signal $u(\cdot)$ that makes the system unsafe, but rather to deactivate the safety filter such that unsafe control actions \textit{could} be passed to the plant. Deactivation of the safety filter is a necessary first step in a two-pronged attack, where the second step would be to modify $u_{\text{des}}$ to make $x \notin \mathcal{X}$.  
\end{remark}

To deactivate the safety filter, we assume that the adversary has access to the following information: 

\begin{assumption}
    \label{as:mod-y}
    The adversary can freely modify $y \leftarrow y^a$. 
\end{assumption}
\begin{assumption}
    \label{as:observ-info}
    The adversary knows the observer gain $K$, the measurement function $h$, and the detector threshold $\delta$.   
\end{assumption}
\begin{assumption}
    \label{as:set-info}
    The adversary can evaluate the function $h_S$ at the current state estimate $\hat{x}$.
\end{assumption}
\begin{assumption}
    \label{as:xhat}
    The adversary knows the current state estimate $\hat{x}$.
\end{assumption}

For satisfying Assumption \ref{as:xhat}, the adversary can either have direct access to the output of the observer, or run its own identical filter in parallel. In the latter case the adversary would need to know the dynamics of the system (i.e., $f$) and the initial state, in addition to $K$ and $h$ from Assumption \ref{as:observ-info}.

In summary, the problem that the adversary is interested in solving is formalized as follows:

\begin{problem}[Problem formulation for adversary]
    \label{prob:adv}
    Under Assumptions \ref{as:mod-y}--\ref{as:xhat}, inject false measurements $y^a$ such that $\hat{x}$ deactivates the ASIF defined in \eqref{eq:qp-cbf}, without being detected by the anomaly detector. 
\end{problem}

In the rest of the paper we consider a specific heuristic attack policy to tackle Problem \ref{prob:adv}.  

\section{Stealthy deactivation of safety filter}
The main idea behind the proposed attack policy is to inject false measurements $y^a$ to bias the state estimate $\hat{x}$ such that the safety filter gets a false sense of safety. To this end, the adversary injects measurements $y^a$ to increase the value of $h_S(\hat{x})$ (recall the $h_S$ defines the safe set $S$ in \eqref{eq:Sdef}), which can be interpreted as increasing the perceived safety margin to the boundary of $S$. As a result, since the safety margin is perceived to be larger than it actually is, the safety filter might erroneously accept dangerous control actions being applied to the plant. 

To be more concrete, by increasing $h_S(\hat{x})$ the set $\tilde{\mathcal{U}}_S(\hat{x})$ of control actions that are perceived to be safe becomes larger since the right-hand-side of the inequality constraint in \eqref{eq:cbf-U} becomes smaller. This incites deactivation of the safety filter.

For the adversary to remain stealthy, the injected measurement $y^a$ needs to be kept close to the expected measurement $h(\hat{x})$ to remain undetected by the anomaly detector; specifically, $\|y^a - h(\hat{x})\|\leq \delta$.

\subsection{Formalization of the attack}
We formalize increasing the safety margin $h_S$ subject to a stealth constraint as finding $y$ that solves  
\begin{equation}
    \begin{aligned}
        \label{eq:grad-attack-nom}
        y^a(\hat{x}) = \argmax_y \:\: & \dot{h}_S(\hat{x}) \\
        \text{subject to } &\|y-h(\hat{x})\| \leq \delta \\
                           & \text{and the dynamics in \eqref{eq:xhat-dot}}.
        \end{aligned}
\end{equation}
The following theorem makes the adversary's attack policy based on \eqref{eq:grad-attack-nom} more explicit.
\begin{theorem}
    Solving \eqref{eq:grad-attack-nom} is equivalent to solving  

\begin{equation}
    \label{eq:grad-attack}
    \begin{aligned}
    y^a(\hat{x}) = \argmax_y\:\: &\nabla h_S(\hat{x})^T K y \\
        \text{subject to } &\|y-h(\hat{x})\| \leq \delta.
    \end{aligned}
\end{equation}
\end{theorem}
\begin{proof}
    Expanding $\dot{h}_S(\hat{x})$ and using \eqref{eq:xhat-dot} yields
    \begin{equation}
        \begin{aligned}
            \dot{h}_S(\hat{x}) &= \nabla h_S(\hat{x})^T \dot{\hat{x}} \\
                               &= \nabla h_S(\hat{x})^T \left(f(\hat{x},u) + K(y-h(\hat{x}))\right).
        \end{aligned}
    \end{equation}
    Since the only term that contains the decision variable $y$ is $\nabla h_S(\hat{x})^T K y$, we can use that term instead of $\dot{h}_S(\hat{x})$ in the objective function and still get the same optimizer. 
\end{proof}

An interpretation of the objective in \eqref{eq:grad-attack} is that $Ky$ is the change to the state due to the measurement $y$. Hence, $\nabla h_S(x)^T K y$ is a measure of how much the change in $x$ aligns with the gradient of the safety margin $h_S$.  

The attack policy obtained through the optimization problem in \eqref{eq:grad-attack} does, in fact, take a closed form when the stealth constraint is expressed in terms of the $2$- or $\infty$-norm.  
\begin{corollary}[Closed-form attack policy for 2-norm]
    \label{cor:2norm}
    If the stealth constraint in \eqref{eq:grad-attack} is posed in terms of the $2$-norm and $\|K^T \nabla h_S(\hat{x})\| \neq 0$, the solution $y^a(\hat{x})$ takes the closed form 
\begin{equation}
    \label{eq:closedform-attack}
    y^a(\hat{x}) = h(\hat{x})+\delta \frac{K^T \nabla h_S(\hat{x})}{\|K^T \nabla h_S(\hat{x})\|_2}.
\end{equation}
\end{corollary}
\begin{proof}
    A proof is provided in Appendix \ref{ap:pf-2norm}.
\end{proof}

\begin{corollary}[Closed-form attack policy for $\infty$-norm]
    \label{cor:infnorm}
    If the stealth constraint in \eqref{eq:grad-attack} is posed in terms of the $\infty$-norm, the solution $y^a(\hat{x})$ takes the closed form 
\begin{equation}
    \label{eq:closedform-attack}
    y^a(\hat{x}) = h(\hat{x})+\delta\sgn(K^T \nabla h_S(\hat{x})),
\end{equation}
where $\sgn$ is evaluated element-wise.
\end{corollary}
\begin{proof}
    A proof is provided in Appendix \ref{ap:pf-infnorm}.
\end{proof}

Both Corollary \ref{cor:2norm} and \ref{cor:infnorm} make it straightforward for the adversary to select a false measurement $y^a$. No matter the norm, however, one can show that the resulting false-data injection results in a positive bias for $\dot{h}$, this is the main motivation behind the proposed attack.
\begin{theorem}
    \label{th:dual-norm}
    Let $\|\cdot\|_*$ denote the dual norm of $\|\cdot\|$, defined as $\|z\|_*=\max_{\|x\|\leq 1} z^T x$ (see, e.g., \cite[$\S$A.1.6]{boyd2004convex} for details.) If the attack policy $y=y^a(\hat{x})$ is used, the change in the safety margin $h_S$ is given by 
    \begin{equation}
        \dot{h}_S(\hat{x})\big|_{y=y^a(\hat{x})} = \nabla h_S(\hat{x})^T f(\hat{x},u) + \delta \|K^T \nabla h_S(\hat{x})\|_*. 
    \end{equation}
\end{theorem}
\begin{proof}
    A proof is provided in Appendix \ref{ap:pf-dualnorm}.
\end{proof}

From Theorem \ref{th:dual-norm}, we see that the attack policy result in a nonnegative bias $\delta\|K^T\nabla h_S(\hat{x})\|$ to $\dot{h}_S(\hat{x})$, which is accordance with the objective of the adversary to make the safety margin $h_S(\hat{x})$ larger (which in turn makes the set $\tilde{\mathcal{U}}_S(\hat{x})$ larger.)

An approach for the attacker to remain as covert as possible, and performing attacks with high impact, is to only perform an attack if the system is operating close to the edge of the safe set. For example, an attack can be initiated if the attack condition $h_S(x) < \gamma$ is satisfied for some user-specified threshold $\gamma$.

We summarize the proposed false-data injection attack in Algorithm~\ref{alg:attack}.
\begin{algorithm}[H]
    \caption{False-data injection attack to deactivate the safety filter in \eqref{eq:qp-cbf}.}
  \label{alg:attack}
  \begin{algorithmic}[1]
      \Require $K$, $h(\cdot)$, $h_S(\cdot)$ 
      \While{attack condition is satisfied}
      \State get $\hat{x}$ (either directly access $\hat{x}$ or run parallel filter) 
      \State $y^a \leftarrow$ formulate and solve \eqref{eq:grad-attack} 
      \State inject $y^a$ to the system
      \EndWhile
  \end{algorithmic}
\end{algorithm}

\subsection{Detections of attack}
The anomaly detector that the attack policy in \eqref{eq:grad-attack} circumvents considers only the \textit{magnitude} of the of the residual $y-h(\hat{x})$. However, the attack policy builds on producing measurements that are biased towards the interior of the safe set $S$. As a result, the direction of the residuals will ``unnaturally'' align with $\nabla h_S$.  A possible counter-measure to the attack is, hence, to detect if the residual aligns with $\nabla h_S$. A metric that measures such an alignment is  
\begin{equation}
    \label{eq:rho}
    \rho(y,\hat{x}) \triangleq  \frac{\nabla h_S(\hat{x})^T K (y-h(\hat{x}))}{\delta \|K^T \nabla h_S(\hat{x})\|}.
\end{equation}
The denominator of \eqref{eq:rho} is a normalization factor to ensure that ${|\rho(y,\hat{x})|\leq 1}$. A large value of $\rho$ signifies that the change in the state $x$ due to the residual $(K(y-h(\hat{x})))$ aligns with $\nabla h_S$, i.e., is biased towards the interior of $S$.  
\begin{corollary}
    \label{cor:rho}
    For the attack policy $y= y^a(\hat{x})$, the correlation measure $\rho$ in \eqref{eq:rho} is 
    \begin{equation}
        \rho(y^a(\hat{x}),\hat{x}) = \frac{\|K^T \nabla h_S(\hat{x})\|_*}{\|K^T \nabla h_S(\hat{x})\|}.
    \end{equation}
\end{corollary}
\begin{proof}
    This is a direct biproduct of the proof of Theorem \ref{th:dual-norm}. Specifically, see \eqref{eq:dual-norm-K}.
\end{proof}

A detector based on \eqref{eq:rho} can, for example, be implemented by thresholding a moving average of $\rho$ as 
\begin{equation}
    \label{eq:ma}
    \frac{1}{T}\int_{t-T}^{t} \rho(y(\tau),\hat{x}(\tau)) d\tau  > \nu,
\end{equation}
with the horizon $T$ and threshold $\nu$.

\section{Numerical Experiments}
\begin{figure}
  \centering
  \begin{tikzpicture}[scale=1]
    \begin{axis}[
        xmin=-1.9,xmax=0.5,
        ymin=0,ymax=1.2,
        xlabel={$x_1$},
        ylabel={$x_2$},
        legend style={at ={(0.5,1.2)},anchor=north}, ymajorgrids,yminorgrids,xmajorgrids,
        y post scale=0.56,
        legend cell align={left},legend columns=2,
        ]
        \draw[dashed,fill=set19c5, opacity=0.2,draw=red] (axis cs:0,0) rectangle (axis cs:4,4);
        \addplot [name path=f, set19c3,dashed, domain = -2:0,samples=1000, unbounded coords=jump]{sqrt(-2*x)};
        \path[name path=axis] (axis cs:-2,0) -- (axis cs:0,0);
        \addplot[set19c3,opacity=0.2] fill between[of=f and axis];

        \addplot [set19c1,very thick] table [x={x1}, y={x2}] {\gradtwodofnew}; 
        \addplot [set19c3,very thick] table [x={z1}, y={z2}] {\gradtwodofnew}; 
        \legend{,,Actual, Perceived}
    \end{axis}
\end{tikzpicture}
  \caption{The actual trajectory $x(t)$ and the perceived trajectory $\hat{x}(t)$ when the false-data injection attack defined by \eqref{eq:grad-attack} is performed.}
  \label{fig:advtrajs}
\end{figure}
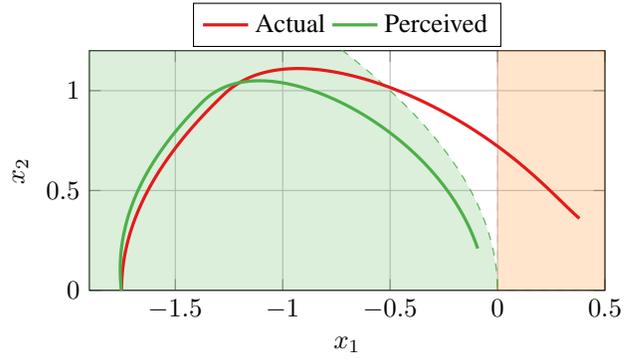
To illustrate the proposed false-data injection attack given in Algorithm~\ref{alg:attack}, we consider the double integrator example in \cite{hobbs2023rta}. Code for all reported experiments is available at {\url{https://github.com/darnstrom/ecc24-sf}}. 

The dynamics of the double integrator system is 
\begin{equation}
  f(x,u) = \begin{pmatrix}
   x_2 \\ u  
  \end{pmatrix}.
\end{equation}
For safety, the admissible states are $\mathcal{X} = \{x \in \mathbb{R}^2 : x_1 \leq 0\}$, and admissible control actions are $\mathcal{U} = \{u\in \mathbb{R} : |u| \leq 1\}$. 

We consider the control invariant set $S $ defined as the zero-superlevel set of $h_S(x) = -2 x_1 - x_2^2$. With the extended class $\mathcal{K}$ function $\alpha(x) = 2x$, the function $h_S$ is a control-barrier function with the corresponding set of safe control actions 
\begin{equation}
    \tilde{\mathcal{U}}_S(x) = \{u\in \mathcal{U}: -2 x_2(1+u) \geq 4 x_1 + 2x_2^2\}.
\end{equation}

For the observer, the gain $K$ is selected as the stationary Kalman gain (obtained by using the process noise $Q=\left(\begin{smallmatrix}
    1 & 0 \\ 0 & 1 
\end{smallmatrix}\right)$ and measurement noise $R = \left(\begin{smallmatrix}
    0.001 & 0 \\ 0 & 0.001 
    \end{smallmatrix}\right)$.) 
    Two scenarios of measurement functions $h$ are consider: either both $x_1$ and $x_2$ are measured (i.e., $h(x) = \left(\begin{smallmatrix}
 x_1 \\ x_2 
\end{smallmatrix}\right)$), or only $x_1$ is measured (i.e., $h(x) = x_1$). In the detector the $2$-norm is used, and its threshold is set to $\delta = 0.001$. 

To activate the safety filter, the desired control action is constantly saturated at ${u_{\text{des}} =1}$, which drives $x_1(t) \to \infty$ as $t\to \infty$ if there would be no safety filter.  

\begin{figure*}
  \centering
  \subfloat[Safety margin \label{subfig:margin}]{%
      \begin{tikzpicture}[scale=1]
    \begin{axis}[
        xmin=0,xmax=3,
        ylabel={$h_S$},
        xlabel={Time [s]},
        legend style={at ={(0.69,0.98)},anchor=north}, ymajorgrids,yminorgrids,xmajorgrids,
        y post scale=0.56,
        legend cell align={left},legend columns=1,
        ]
        \addplot [set19c2,very thick] table [x={t}, y={hx}] {\randomnew}; 
        \addplot [set19c4,very thick] table [x={t}, y={hx}] {\gradonedofnew}; 
        \addplot [set19c1,very thick] table [x={t}, y={hx}] {\gradtwodofnew}; 
        \legend{Random, Actual ($h(x)=x_1$), Actual ($h(x)= \left(\begin{smallmatrix}
         x_1 \\ x_2 
\end{smallmatrix}\right)$)}
    \end{axis}
\end{tikzpicture}
  }
  \hfill
  \subfloat[Trajectories \label{subfig:trajs}]{%
      \begin{tikzpicture}[scale=1]
    \begin{axis}[
        xmin=-1.9,xmax=0.5,
        ymin=0,ymax=1.2,
        xlabel={$x_1$},
        ylabel={$x_2$},
        legend style={at ={(0.5,1.2)},anchor=north}, ymajorgrids,yminorgrids,xmajorgrids,
        y post scale=0.56,
        legend cell align={left},legend columns=3,
        ]
        \draw[dashed,fill=set19c5, opacity=0.2,draw=red] (axis cs:0,0) rectangle (axis cs:4,4);
        \addplot [name path=f, set19c3,dashed, domain = -2:0,samples=1000, unbounded coords=jump]{sqrt(-2*x)};
        \path[name path=axis] (axis cs:-2,0) -- (axis cs:0,0);
        \addplot[set19c3,opacity=0.2] fill between[of=f and axis];

        \addplot [set19c2,very thick] table [x={x1}, y={x2}] {\randomnew}; 
        \addplot [set19c4,very thick] table [x={x1}, y={x2}] {\gradonedofnew}; 
        \addplot [set19c1,very thick] table [x={x1}, y={x2}] {\gradtwodofnew}; 
    \end{axis}
\end{tikzpicture}
}
\caption{The safety margin and resulting state trajectories when an attack with random directions according to \eqref{eq:random-attack} is occurring, and when an adversary performs false-data injection attacks according to \eqref{eq:grad-attack} with  $h(x) = x_1$ and $h(x)=\left(\begin{smallmatrix}
 x_1 \\ x_2 
\end{smallmatrix}\right)$, respectively.}
\end{figure*}
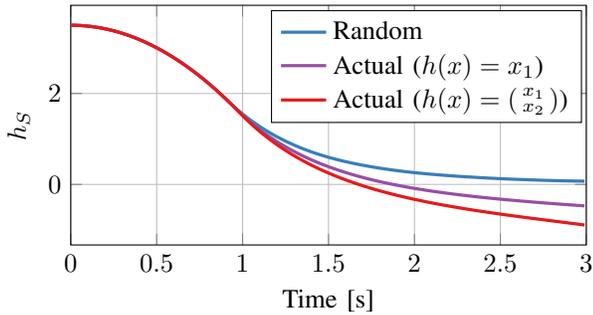
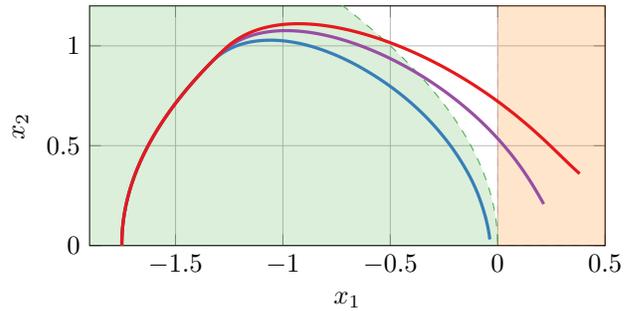

For the implementation of the safety filter, the quadratic programming solver DAQP \cite{arnstrom2022daqp} is used to solve optimization problems of the form \eqref{eq:grad-attack}.

The attack is active for the entire duration of the simulation. Since the $2$-norm is used in the detector we have that  ${y(t) = h(\hat{x}(t))+\delta \frac{K^T \nabla h_S(\hat{x}(t))}{\|K^T \nabla h_S(\hat{x}(t))\|_2}}$ for all $t\geq 0$ from the closed-form expression in Corollary~\ref{cor:2norm}. 

First we consider an attack when both states are measured. The system is started in $x(0)= \left(\begin{smallmatrix}
 -1.75 \\ 0 
\end{smallmatrix}\right)$ and is simulated for $3$ seconds. Figure \ref{fig:advtrajs} illustrates the resulting \textit{perceived} trajectory (i.e., $\hat{x}(t)$) and the \textit{actual} trajectory of the system (i.e., $x(t)$). The green region is the safe set $S$ and the red region marks the inadmissible states. As can be seen, the false-data injection makes the observer believe that the states remain within the safe set, while in actuality they leave it and pass over to the inadmissible states; the adversary's objective of making the system unsafe is, hence, achieved.

Next, we consider the same setup except that only the first state is measured (i.e., $h(x) = x_1$). The resulting trajectory $x(t)$ is shown in Figure \ref{subfig:trajs} together with the trajectory from the first scenario and for when a stealthy attack with random directions of the form 
\begin{equation}
    \label{eq:random-attack}
y=h(\hat{x})+\delta \frac{e}{\|e\|_2},\qquad e \sim \mathcal{N}(0,1),
\end{equation}
is performed. The corresponding safety margin over the simulation is shown in Figure \ref{subfig:margin}. The results illustrate that the safety filter remains active when a naive random attack of the form \eqref{eq:random-attack} is performed. In contrast, the attack in \eqref{eq:grad-attack} deactivates the safety filter, and the safe set is exited quicker when both states are measured, which is expected since then the adversary is given more degrees of freedom in the attack. 

Finally, we illustrate how the correlation measure $\rho(y,\hat{x})$ in \eqref{eq:rho} can expose the proposed attack. Figure \ref{fig:rho} shows $\rho$ from the first scenario when both states are measured and the attack is active throughout the entire simulation. In addition, Figure \ref{fig:rho} shows $\rho$ for when an attack with random directions according to \eqref{eq:random-attack} is performed. Note that, in both instances, the attacks are stealthy w.r.t. the residual $y-h(\hat{x})$. In contrast to a random attack, $\rho$ is constantly $1$ for the attack from \eqref{eq:grad-attack} (as is expected from Corollary~\ref{cor:rho} since the dual norm of the 2-norm is also the 2-norm.) Figure \ref{fig:rhoma} shows a corresponding moving average of $\rho$ according to \eqref{eq:ma} (with a horizon of $T=0.25$ seconds) for the two attacks. For a threshold of, e.g., $\nu = 0.9$, an attack that biases state estimates toward the interior of $S$ is correctly detected when the attack in \eqref{eq:grad-attack} is active, while the random attack correctly remains undetected. 

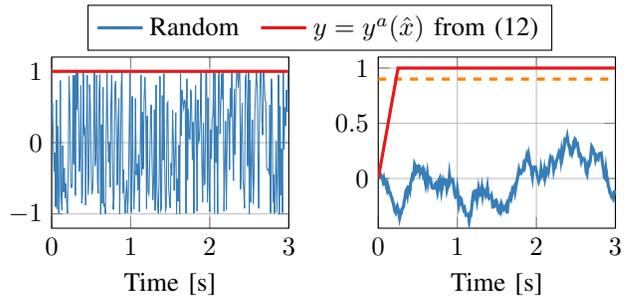
\begin{figure}
  \centering
  \begin{tikzpicture}
      \begin{axis}[%
          hide axis,
          xmin=0, xmax=1,
          ymin=0, ymax=1,
          legend style={legend cell align=left, legend columns=2},
          legend style={/tikz/every even column/.append style={column sep=0.2cm}}
          ]
          \addlegendimage{set19c2,very thick}
          \addlegendimage{set19c1,very thick}
          \legend{Random, $y=y^a(\hat{x})$ from \eqref{eq:grad-attack}}
      \end{axis}
  \end{tikzpicture}
  \subfloat[Correlation $\rho(y,\hat{x})$ \label{fig:rho}]{%
      \begin{tikzpicture}[scale=1]
    \begin{axis}[
        xmin=0,xmax=3,
        xlabel={Time [s]},
        legend style={at ={(0.5,1.3)},anchor=north}, ymajorgrids,yminorgrids,xmajorgrids,
        x post scale=0.46,
        y post scale=0.4,
        legend cell align={left},legend columns=2,
        ]
        \addplot [set19c2] table [x={t}, y={rho}] {\randomnew}; 
        \addplot [set19c1,very thick] table [x={t}, y={rho}] {\gradtwodofnew}; 
    \end{axis}
\end{tikzpicture}
  }
  \hfill
  \subfloat[Moving average ($T=0.25$) \label{fig:rhoma}]{%
      \begin{tikzpicture}[scale=1]
    \begin{axis}[
        xmin=0,xmax=3,
        ymin=-0.45,ymax=1.1,
        xlabel={Time [s]},
        legend style={at ={(0.5,1.3)},anchor=north}, ymajorgrids,yminorgrids,xmajorgrids,
        x post scale=0.46,
        y post scale=0.4,
        legend cell align={left},legend columns=2,
        ]
        \addplot[mark=none,very thick, dashed, set19c5,domain=0:5] {0.9};
        \addplot [set19c2,very thick] table [x={t}, y={MA}] {\randomnew}; 
        \addplot [set19c1,very thick] table [x={t}, y={MA}] {\gradtwodofnew}; 
    \end{axis}
\end{tikzpicture}
}
\caption{The correlation measure $\rho$ defined in \eqref{eq:rho} and a moving average according to \eqref{eq:ma} under two different attacks: a stealthy false-data injection attack with random directions according to \eqref{eq:random-attack}, and a stealthy false-data injection attack according to  \eqref{eq:grad-attack}. An example threshold (\textcolor{set19c5}{$\nu = 0.9$}) for a detector of the form \eqref{eq:ma} is shown as a dashed line, which would detect the attack in \eqref{eq:grad-attack} after about 0.2 seconds.}
\label{fig:detector}
\end{figure}

\section{Conclusion}
We have proposed a simple and stealthy false-data injection attack for deactivating CBF-based safety filters. The attack injects false sensor measurements to bias state estimates to the interior of a safety region, which makes the safety filter accept unsafe control actions. We have also proposed a detector that detects biases manufactured by the proposed attack policy, which complements conventional detectors when safety filters are used. 
We have illustrated that an adversary can successfully make a system unsafe by performing the proposed false-data injection attack on a double integrator system. Moreover, we have shown with the same example that the proposed detector can be used to detect when such a false-data injection attack is happening.

Future work includes considering similar attacks but with less information required by the adversary; for example by considering that only the observations $y$ are available to the adversary. We will also explore if a receding horizon attack can be more effective that the attack considered herein. 

\bibliographystyle{IEEEtran}
\linespread{0.9}\selectfont
\bibliography{lib.bib}

\appendix
\subsection{Proof of Corollary \ref{cor:2norm}}
\label{ap:pf-2norm}
\begin{proof}
    First, note that the constraint ${\|y-h(\hat{x})\|_2 \leq \delta}$ can be equivalently stated in the differentiable form ${\|y-h(\hat{x})\|^2_2 \leq \delta^2}$.  The KKT-conditions for \eqref{eq:grad-attack} with the 2-norm constraint are then given by 
    \begin{subequations}
        \begin{align}
            \gradxhat + \lambda(y^*-h(\hat{x})) &= 0 \label{pf-cor1-stat}\\
            \|y^*-h(\hat{x})\|_2 &\leq \delta \\
            \lambda &\geq 0  \\
            \lambda (\delta-\|y^*-h(\hat{x})\|_2) &= 0 \label{pf-cor1-comp}
        \end{align}
    \end{subequations}
    If $\gradxhat \neq 0$ we have that $\lambda \neq 0$ from the stationarity condition \eqref{pf-cor1-stat}. Hence, we can rewrite \eqref{pf-cor1-stat} as  
\begin{equation}
    \label{eq:pf-cor1-stat}
    y^*-h(\hat{x}) = \frac{\gradxhat}{\lambda}.
\end{equation}
Since $\lambda \neq 0$, the complementarity condition \eqref{pf-cor1-comp} implies that $\|y^*-h(\hat{x})\|_2 = \delta$. Taking the 2-norm of both sides of \eqref{eq:pf-cor1-stat} therefore gives  
\begin{equation}
    \label{eq:pf-cor1-final}
    \delta = \frac{\|\gradxhat\|_2}{\lambda} \Leftrightarrow \lambda = \frac{\|\gradxhat\|_2}{\delta}.
\end{equation}
Inserting \eqref{eq:pf-cor1-final} into \eqref{eq:pf-cor1-stat} gives
\begin{equation}
    y^* = h(\hat{x})+\delta \frac{\gradxhat}{\|\gradxhat\|_2}.
\end{equation}
\end{proof}
\subsection{Proof of Corollary \ref{cor:infnorm}}
\label{ap:pf-infnorm}
\begin{proof}
First, note that the constraint $\|y-h(\hat{x})\|_{\infty} \leq \delta$ can be split into the linear inequalities $y-h(\hat{x}) \leq \delta$ and $-y+h(\hat{x})\leq \delta$. The KKT-conditions for \eqref{eq:grad-attack} with the $\infty$-norm constraint are then given by 
\begin{subequations}
    \begin{align}
        \gradxhat &=-\lambda^+ +\lambda^-,  \label{pf-cor2-stat}\\
        y-h(\hat{x}) &\leq \delta, \quad -y+h(\hat{x}) \leq \delta\\
        \lambda^+ &\geq 0,\quad \lambda^-\geq 0 \\
        [\lambda^+]_i [\delta-h(\hat{x})+y]_i  &= 0,\quad [\lambda^-]_i [-\delta+h(\hat{x})-y]_i \label{pf-cor2-comp}.
    \end{align}
\end{subequations}
If $[\gradxhat]_i > 0$, we have that $[\lambda^+]_i > 0$ for the stationarity condition in \eqref{pf-cor2-stat} to hold. Similarly, $[\gradxhat]_i < 0$ requires $[\lambda^-] > 0$. This together with the complementary conditions \eqref{pf-cor2-comp} gives 
\begin{equation}
    [y]_i = \begin{cases}
        h(\hat{x})+\delta, &\text{ if } [\gradxhat]_i > 0,\\
        h(\hat{x})-\delta, &\text{ if } [\gradxhat]_i < 0,\\
    \end{cases}  
\end{equation}
which can be compactly written as 
\begin{equation}
    y = h(\hat{x})+\delta \sgn(\gradxhat).
\end{equation}
\end{proof}
\subsection{Proof of Theorem \ref{th:dual-norm}}
\label{ap:pf-dualnorm}

\begin{proof}
    First we expand the expression for $\dot{h}(\hat{x})$ and insert $y=y^a(\hat{x})$ from \eqref{eq:closedform-attack}, which gives  
    \begin{equation*}
      \begin{aligned}
          \dot{h}_S(\hat{x}) &= \nabla h_S(\hat{x}) \dot{\hat{x}}  \\  
                             &=  \nabla h_S(\hat{x})^T\left(f(\hat{x},u)+K\left(y^a(\hat{x})-h(\hat{x})\right)\right) \\
                             &=  \nabla h_S(\hat{x})^T f(\hat{x},u) + \nabla h_S(\hat{x})^TK\left(y^a(\hat{x})-h(\hat{x})\right).
      \end{aligned}
    \end{equation*}
    
    Next, we intend to rewrite the second term $\nabla h_S(\hat{x})^TK\left(y^a(\hat{x})-h(\hat{x})\right)$.
    From the definition of $y^a(\hat{x})$ in \eqref{eq:grad-attack} we have
\begin{equation*}
    \begin{aligned}
        \nabla h_S(\hat{x})^T K y^a(\hat{x}) &= \max_{y:\|y-h(\hat{x})\|\leq \delta}\nabla h_S(\hat{x})^T K y \\
                                             &= \max_{\tilde{y} : \|\tilde{y}\| \leq 1} \nabla h_S(\hat{x})^T K (\delta \tilde{y}-h(\hat{x})) \\
                                             &= \max_{\tilde{y} : \|\tilde{y}\| \leq 1} \delta \nabla h_S^T K  \tilde{y} - \nabla h_S^T K h(\hat{x}) \\
                                             &= \|\delta K^T \nabla h_S(\hat{x}) \|_*  - \nabla h_S(\hat{x})^T K h(\hat{x}),
    \end{aligned}
\end{equation*}
where the variable change $\tilde{y} = \frac{1}{\delta}(y-h(\hat{x}))$ has been made in the second equality and the definition of $\|\cdot\|_*$ has been used in the last equality. Equivalently, we then have
\begin{equation}
    \label{eq:dual-norm-K}
    \nabla h_S(\hat{x})^T K (y^a(\hat{x})-h(\hat{x})) = \|\delta K^T \nabla h_S(\hat{x}) \|_*.
\end{equation}
Inserting this in the expression of $\dot{h}_S(\hat{x})$ gives the desired result.
\end{proof}
\end{document}